\documentclass[conference,letterpaper]{IEEEtran}
\makeatletter
\newcommand{\hathat}[1]{%
\begingroup%
  \let\macc@kerna\z@%
  \let\macc@kernb\z@%
  \let\macc@nucleus\@empty%
  \hat{\mathchoice%
    {\raisebox{.2ex}{\vphantom{\ensuremath{\displaystyle #1}}}}%
    {\raisebox{.2ex}{\vphantom{\ensuremath{\textstyle #1}}}}%
    {\raisebox{.16ex}{\vphantom{\ensuremath{\scriptstyle #1}}}}%
    {\raisebox{.14ex}{\vphantom{\ensuremath{\scriptscriptstyle #1}}}}%
    \smash{\hat{#1}}}%
\endgroup%
}
\makeatother

\usepackage[utf8]{inputenc} 
\usepackage[T1]{fontenc}
\usepackage{url}
\usepackage{cite}
\usepackage[cmex10]{amsmath} 
\usepackage{authblk}
\usepackage{hyperref}
\usepackage{algorithm}
\usepackage{algorithmic}
\usepackage{amsthm,amssymb,amsfonts} 
\usepackage{algorithmic}
\usepackage{graphicx}
\usepackage{textcomp}
\usepackage{xcolor}
\usepackage{tikz}
\usepackage{ifthen}
\usepackage[%
  left=0.64in,
  right=0.63in,
  top=0.68in,
  bottom=1.06in,  
]{geometry}

\newcommand{\cA}{\mathcal{A}}

\newcommand{\boldA}{\mathbf{A}}
\newcommand{\boldB}{\mathbf{B}}
\newcommand{\boldC}{\mathbf{C}}

\newcommand{\boldS}{\mathbf{S}}

\newcommand{\boldX}{\mathbf{X}}

\newtheorem{pid axiom}{PID Axiom}
\newtheorem{sid axiom}{SID Axiom}

\newtheorem{remark}{Remark}
\newtheorem{definition}{Definition}
\newtheorem{lemma}{Lemma}
\newtheorem{corollary}{Corollary}
\newtheorem{property}{Property}
\newtheorem{theorem}{Theorem}
\newtheorem{observation}{Observation}

\newboolean{LongVersion}
\setboolean{LongVersion}{false} 

\begin{document}
\title{The Whole Is Less than the Sum of Parts: Subsystem Inconsistency in Partial Information Decomposition}
\author{\textbf{Aobo Lyu}$^\star$, \textbf{Andrew Clark}$^{\star}$, and \textbf{Netanel Raviv}$^\dagger$\\
$^\star$Department of Electrical and Systems Engineering, Washington University in St. Louis, St. Louis, MO, USA\\
		$^\dagger$Department of Computer Science and Engineering, Washington University in St. Louis, St. Louis, MO, USA\\
   \texttt{aobo.lyu@wustl.edu}, \texttt{andrewclark@wustl.edu}, \texttt{netanel.raviv@wustl.edu}}


\maketitle
\makeatletter
\def\@author{}
\makeatother

\begin{abstract}
Partial Information Decomposition (PID) was proposed by Williams and Beer in 2010 as a tool for analyzing fine-grained interactions between multiple random variables, and has since found numerous applications ranging from neuroscience to privacy.
However, a unified theoretical framework remains elusive due to key conceptual and technical challenges. 
We identify and illustrate a crucial problem: PID violates the set-theoretic principle that the whole equals the sum of its parts (WESP). 
Through a counterexample in a three-variable system, we demonstrate how such violations naturally arise, revealing a fundamental limitation of current lattice-based PID frameworks.
To address this issue, we introduce a new axiomatic framework, termed System Information Decomposition (SID), specifically tailored for three-variable systems. SID resolves the WESP violation by redefining the summation rules of decomposed information atoms based on synergistic relationships. 
However, we further show that for systems with four or more variables, no partial summation approach within the existing lattice-based structures can fully eliminate WESP inconsistencies.
Our results thus highlight the inherent inadequacy of (antichain) lattice-based decompositions for general multivariate systems.

\end{abstract}
\pagestyle{empty}

\section{Introduction}
Understanding how information is shared among multiple variables is a fundamental challenge in information theory. 
Partial Information Decomposition (PID), introduced by Williams and Beer~\cite{williams2010nonnegative}, addresses this by decomposing the mutual information between a group of source variables and a target variable into distinct and independent information atoms,
which reveal how the sources individually and collectively influence a target. 
The implementation of this framework is based on a structure called  \textit{the redundancy lattice}~\cite{crampton2001completion}, which is inspired by and closely aligned with the principles of classical set theory.
\ifthenelse{\boolean{LongVersion}}{In brain network analysis, PID (or similar ideas) has been instrumental in measuring correlations between neurons~\cite{schneidman2003synergy} and understanding complex neuronal interactions in cognitive processes \cite{varley2023partial}. For privacy and fairness studies, the synergistic concept provides insights about data disclosure mechanisms~\cite{rassouli2019data,hamman2023demystifying}. In the field of causality, information decomposition can be used to distinguish and quantify the occurrence of causal emergence \cite{rosas2020reconciling}, and more.}{Since its conception PID has been widely applied, providing deep insights into neural correlations~\cite{schneidman2003synergy} and cognitive processes~\cite{varley2023partial}, privacy and fairness in data disclosure~\cite{rassouli2019data,hamman2023demystifying}, causal emergence phenomena~\cite{rosas2020reconciling}, and more.}

However, despite these successful applications, the PID framework itself remains incomplete and faces fundamental challenges. 
Notably, despite extensive research efforts~\cite{griffith2014intersection, ince2017measuring, bertschinger2013shared, harder2013bivariate, bertschinger2014quantifying, lyu2024explicit}, no existing PID measure simultaneously satisfies all PID axioms and desired properties. 
We argue that this persistent difficulty arises primarily from conceptual flaws inherent to the PID framework. 
Specifically, PID implicitly relies on the assumption that \textit{the whole equals the sum of its parts (WESP)}—a set-theoretic principle, which might be violated by information measures. 
Such violations have previously been noted~\cite{kolchinsky2022novel} through the perspective of the inclusion-exclusion principle~\cite{ting1962amount,yeung1991new}. 
In this paper, we point out that enforcing WESP overlooks higher-order synergistic effects, resulting in inconsistent behavior in subsystems, where the sum of PID components differs from the total information of the system.

In light of these challenges, there is a need for a new framework that can explore the rules that information decomposition should follow without relying on the WESP principle. 
In this direction, this paper makes three main contributions:
(i) We reveal an inherent inconsistency in PID by explicitly demonstrating WESP violation with a three-source counterexample.
(ii) We introduce System Information Decomposition (SID), a novel axiomatic framework for three-source systems with a target equal to sources, resolving PID’s WESP inconsistency by redefining information summation rules within a simplified antichain structure. 
Then, it is shown that an extension of the well-known Gács-Körner common information measure~\cite{gacs1973common} satisfies this axiomatic framework.
(iii) We demonstrate that for a general multivariate system, no antichain lattice-based summation method can fully resolve the WESP inconsistency, by presenting two systems with identical information atoms, but different mutual information due to different synergistic dynamics.
Our results advocate for a shift beyond antichain (set-partition) logic toward new structural foundations capable of capturing higher-order synergistic relationships.
The remainder of the paper is structured as follows. Section~\ref{sec:PID} reviews PID and presents our counterexample illustrating WESP violations. Section~\ref{sec:SID} proposes the SID axiomatic framework for a three-variable system and an operational definition of redundancy, demonstrating how SID resolves these inconsistencies. Section~\ref{sec:limitation of ac} extends the analysis to four-variable systems, showing why the antichain-based approach is insufficient. Finally, Section~\ref{sec:discussion} discusses the implications of the work and future directions.

\section{Partial Information Decomposition}
\label{sec:PID}
Williams and Beer~\cite{williams2010nonnegative} introduced the Partial Information Decomposition (PID) framework to decompose multivariate information axiomatically. Consider sources~$S_1,S_2$ and target~$T$: the mutual information~$I(S_1,S_2;T)$ decomposes into redundant, unique, and synergistic atoms (see Figure~\ref{fig:PID}). The redundancy $\operatorname{Red}(S_1,S_2\to T)$ is shared information; the unique atom $\operatorname{Un}(S_1\to T|S_2)$ is information exclusively from~$S_1$ (sim. $\operatorname{Un}(S_2\to T|S_1)$), and synergy $\operatorname{Syn}(S_1,S_2\to T)$ emerges only from joint observation of $S_1$ and $S_2$.

\begin{figure}[htbp]
\centering
\fbox{\includegraphics[width=.65\linewidth]{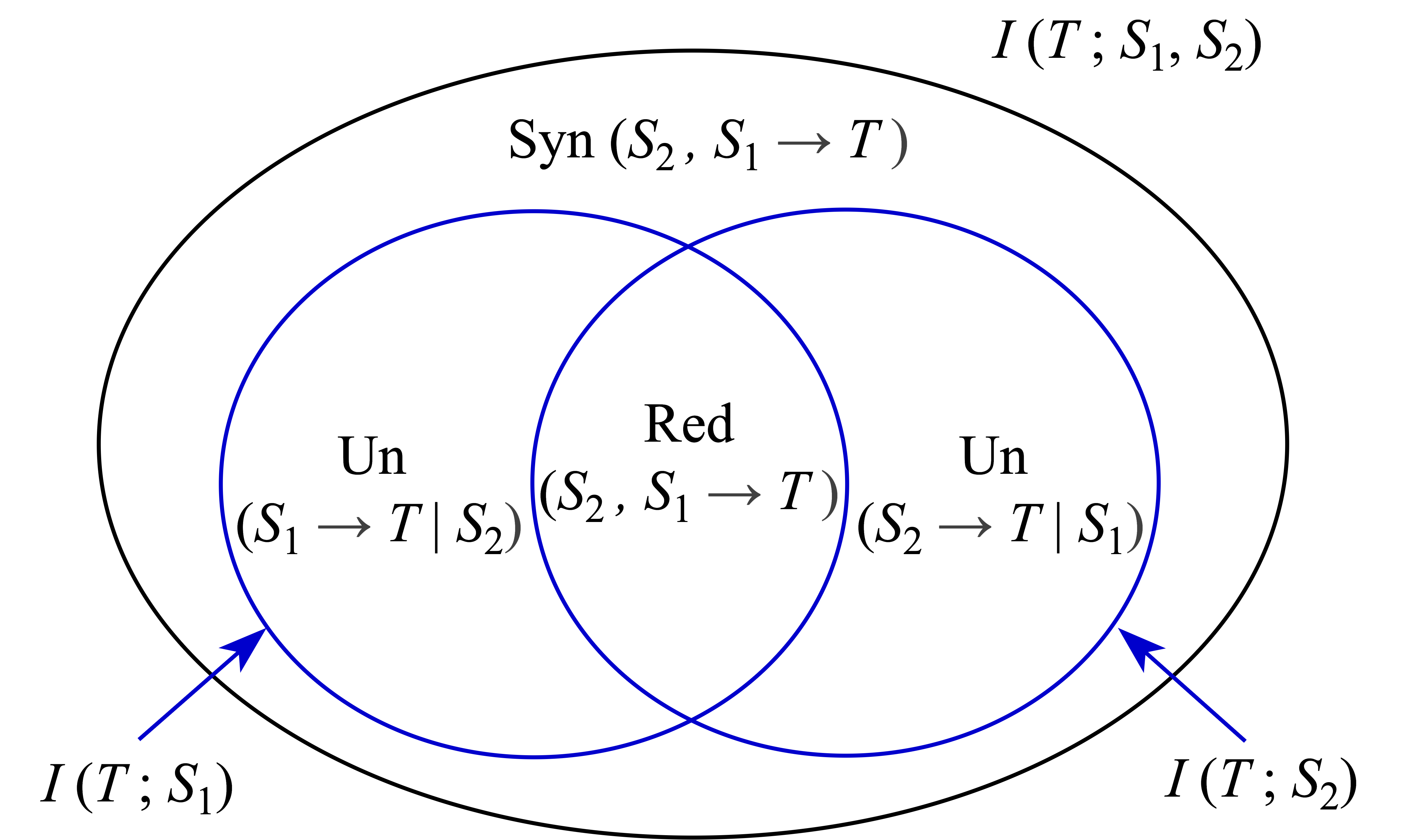 }}
\caption{The structure of PID with two source variables, i.e.,~\eqref{equ:Information Atoms' relationship_1}~\eqref{equ:Information Atoms' relationship_2}.}
\label{fig:PID}
\end{figure}
Together, we have that for the whole system $(S_1,S_2,T)$,
\begin{align}
\label{equ:Information Atoms' relationship_1}
    I((S_1,S_2);T) = \operatorname{Red}(S_1,S_2 \to T) +\operatorname{Syn}(S_1,S_2 \to T) \nonumber\\+\operatorname{Un}(S_1\to T\vert S_2)+ \operatorname{Un}(S_2\to T\vert S_1),
\end{align}
and for each subsystem $(S_1,T)$ and $(S_2,T)$,
\begin{align}\label{equ:Information Atoms' relationship_2}
    I(S_1;T) &= \operatorname{Red}(S_1,S_2 \to T) + \operatorname{Un}(S_1\to T\vert S_2), \mbox{ and}\nonumber\\
    I(S_2;T) &= \operatorname{Red}(S_1,S_2 \to T) + \operatorname{Un}(S_2\to T\vert S_1).
\end{align}

For general systems with source variables $\boldS = \{S_1,\dots,S_n\}$ and target $T$, PID uses \textit{the redundancy lattice}~$\mathcal{A}(\boldS)$~\cite{williams2010nonnegative,crampton2001completion}, which is the set of antichains formed from the power set of $\boldS$ under set inclusion with a natural order~$\preceq_\boldS$.
\begin{definition}[PID Redundancy Lattice]
\label{definition:PID lattice}
For the set of source variables $\boldS$, the set of antichains is:
\begin{align}
\mathcal{A}(\boldS) = \{\alpha \in \mathcal{P}_1(\mathcal{P}_1(\boldS)):\forall \boldA_i,\boldA_j \in \alpha, \boldA_i \not \subset \boldA_j\},
\end{align}
where $\mathcal{P}_1(\boldS) = \mathcal{P}(\boldS) \setminus \{\varnothing\}$ is the set of all nonempty subsets of $\boldS$, and 
for every~$\alpha,\beta\in\mathcal{A}(\boldS)$, we say that~$\beta \preceq_\boldS \alpha$ if for every~$\boldA\in\alpha$ there exists~$\boldB\in \beta$ such that~$\boldB\subseteq \boldA$.
\end{definition}
For ease of exposition, we denote elements of~$\mathcal{A}(\boldS)$ using their indices (e.g., write $\bigl\{\{S_1\}\{S_2\}\bigl\}$ as $\bigl\{\{1\}\{2\}\bigl\}$).
Based on PID~Definition~\ref{definition:PID lattice}, the value of the PI-atoms can be expressed as a partial information function $\Pi^{T}_{\boldA}$.

\begin{definition} [Partial Information Decomposition Framework]
\label{pid def:PIDF}
Let~$\boldS$ be a collection of sources and let~$T$ be the target.
The set of PI-atoms is defined as a family of partial information functions (PI-function) $\Pi^{T}_{\boldA}:\mathcal{A}(\boldA) \rightarrow \mathbb{R}$ for all $\boldA \subseteq \boldS$.  
\end{definition}
Intuitively, for every~$\alpha\in\cA(\boldA)$, the atom $\Pi^{T}_{\boldA}(\alpha)$ measures the amount of information provided by each set in the anti-chain~$\alpha$ to~$T$ and is not provided by any~$\beta \preceq_\boldA \alpha$.

For simplicity we denote $\Pi^{T}_{i\dots}(\cdot)$ for $\Pi^{T}_{\{S_i\dots\}}(\cdot)$, e.g., $\Pi^{T}_{12}(\{\{1\}\})=\Pi^{T}_{\{S_1,S_2\}}(\{\{S_1\}\})$.
Note that in the case~$\boldS=\{S_1,S_2\}$, Definition~\ref{pid def:PIDF} reduces to
\begin{align*}
    \operatorname{Red}(S_1,\!S_2\!\to\!T)\! & =\! \Pi^{T}_{12}(\!\bigl\{\!\{\!1\!\}\!\{\!2\!\}\!\bigl\}\!), \operatorname{Un}(S_1\!\to\! T\vert S_2\!) \!=\! \Pi^{T}_{12}(\!\bigl\{\!\{1\}\!\bigl\}\!),\\
    \operatorname{Syn}(S_1,S_2 \!\to\! T)\! &=\! \Pi^{T}_{12}(\!\bigl\{\!\{12\}\!\bigl\}\!), 
    \operatorname{Un}(S_2\!\to\! T\vert S_1)\! =\! \Pi^{T}_{12}(\!\bigl\{\!\{2\}\!\bigl\}\!),
\end{align*}
recovering the terms in~\eqref{equ:Information Atoms' relationship_1} and~\eqref{equ:Information Atoms' relationship_2}.
\begin{figure}[htbp]
\centering
\fbox{\includegraphics[width=.75\linewidth]{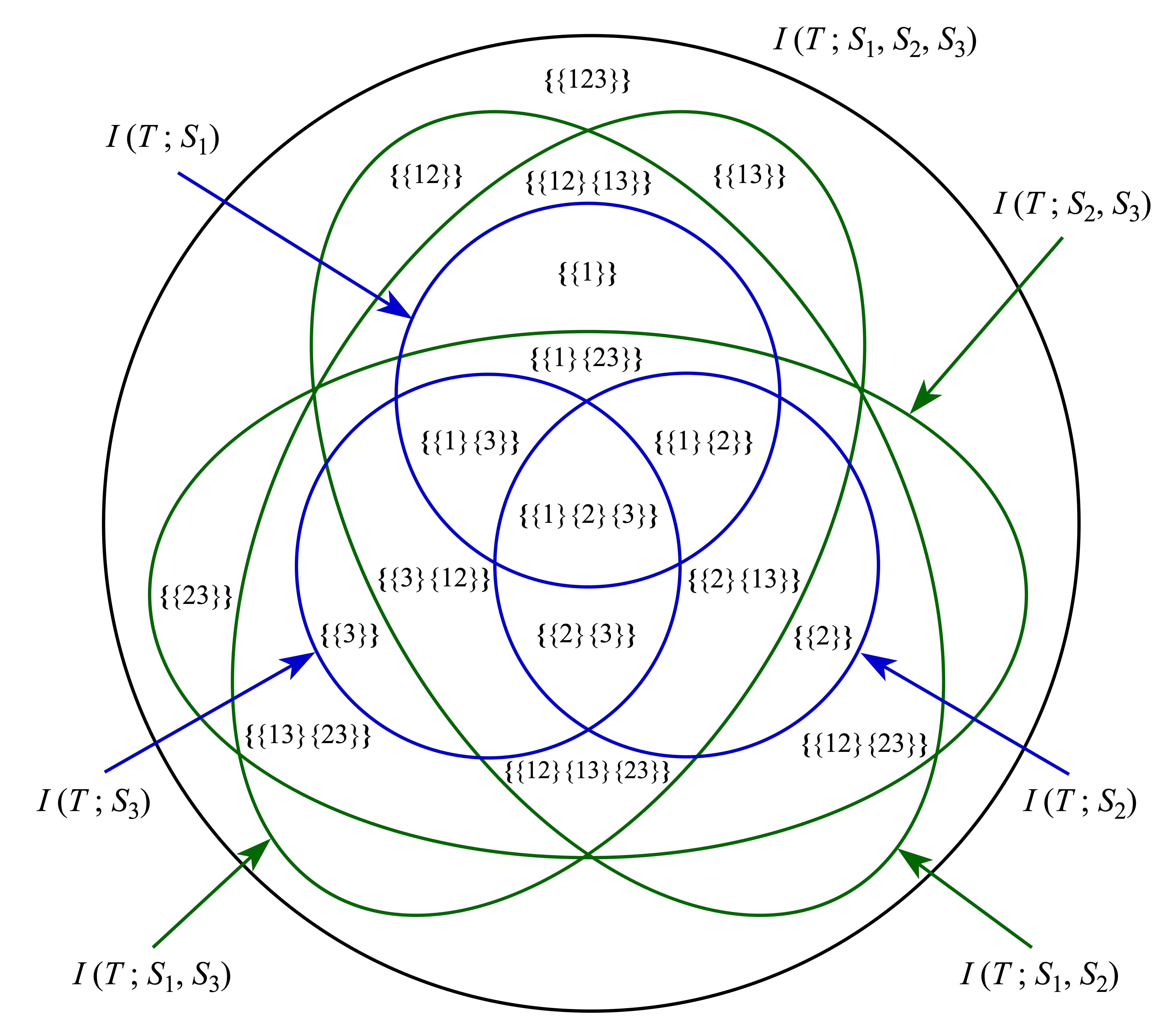 }}
\caption{The structure of PID with 3 source variables. }
\label{fig:ThreeCase}
\end{figure}

For general systems, PID requires the following mutual information constraints~\cite{williams2010nonnegative} (i.e., the equivalent of~\eqref{equ:Information Atoms' relationship_1} and~\eqref{equ:Information Atoms' relationship_2}).

\begin{pid axiom}
\label{pid axiom:mutual constrains}
    For any subsets $\boldA$, $\boldB$ of sources~$\boldS$ with $\boldA\subseteq \boldB$, the sum of PI-atoms decomposed from system $\boldB$ satisfies
\begin{align}
\label{equ:PID Information Atoms}
    I(\boldA;T) = \sum_{\beta \preceq_{\boldB} \{\boldA\} } \Pi^{T}_{\boldB}(\beta),
\end{align}
where $\{\boldA\}$ is the antichain with a single element~$\boldA$.
\end{pid axiom}
It is worth noting that \eqref{equ:PID Information Atoms} constrains the consistency of the sum of PI-atoms decomposed from different subsystems \cite{ince2017partial,chicharro2017synergy,rosas2020operational,lizier2013towards} by taking different sets $\boldB$ on the right-hand side.
\begin{lemma}[Subsystem Consistency]
\label{lemma: subsystem consistency}
For~$\boldA,\boldB,\boldC\subseteq \boldS$ such that~$\boldC\subseteq\boldA\cap \boldB$, let $\Pi$ from Definition~\ref{pid def:PIDF} that satisfying PID Axiom~\ref{pid axiom:mutual constrains}, we have that
    \begin{align}
    \label{equ:subsystem}
        \sum_{\beta \preceq_\boldA \{\boldC\}} \Pi^{T}_{\boldA}(\beta)=\sum_{\beta \preceq_\boldB \{\boldC\}} \Pi^{T}_{\boldB}(\beta).
    \end{align}    
\end{lemma}
Consider the system in Figure~\ref{fig:PID}.
For the atoms decomposed from the system $(S_1,T)$, the quantity $\Pi^{T}_{1}(\bigl\{\{1\}\bigl\})$ reflects the (redundant) information that $S_1$ provides about~$T$.
If we add a source~$S_2$ to this system, this information will be further decomposed into the redundant information from $S_1,S_2$ and the unique information only from $S_1$ but not $S_2$. 
Below are three axioms regarding the redundant information $\operatorname{Red}(S_1,\dots,S_N\to T)$---which is reflected by the PI-atom $\Pi^{T}_{\boldS}(\bigl\{\{1\}\dots\{N\}\bigl\})$---for any multivariate system~$\boldS$.
\begin{pid axiom} [Commutativity\footnote{We note that in general, commutativity is an axiom that follows from the intuitive understanding of redundancy. 
Yet, when using the anti-chain perspective, it is trivially satisfied since the respective atom~$\Pi^T_{\boldS}(\{\{1\}\{2\}\ldots\{N\}\})$ depends on the anti-chain of all singletons, which is not affected by reordering.}]
\label{PID Axiom: Commutativity}
Redundant information is invariant under any permutation $\sigma$ of sources, i.e., $\operatorname{Red}(S_1,\dots,S_N\to T) = \operatorname{Red}(S_{\sigma(1)}, \dots, S_{\sigma(N)}\to T)$.
\end{pid axiom}  



\begin{pid axiom} [Monotonicity]
\label{pid axiom: Monotonicity}
Redundant information decreases monotonically as more sources are included, i.e., 
$\operatorname{Red}(S_1,\dots,S_{N},S_{N+1} \to T) \le \operatorname{Red}(S_1,\dots,S_{N} \to T)$.
\end{pid axiom}

\begin{pid axiom} [Self-redundancy]
\label{pid axiom: Self-redundancy}
Redundant information for a single source variable $S_i$ equals the mutual information, i.e., 
$\operatorname{Red}(S_i \to T) = I(S_i;T).$
\end{pid axiom}
Axiom~\ref{pid axiom: Monotonicity} also implies another lemma, as follows.

\begin{lemma} [Nonnegativity]
\label{Lemma: Nonnegativity}
Partial Information Decomposition satisfies $\operatorname{Red}(S_1,\dots,S_N \to T) \ge 0$.
\begin{proof}
Add a constant variable $S^*$
to the sources and obtain $\operatorname{Red}(\boldA \to T) \ge \operatorname{Red}(\boldA,S^* \to T) =0$. 
\end{proof}
\end{lemma}

Besides, another intuitive property is often considered \cite{ince2017measuring}.
\begin{property}[Independent Identity]
\label{property: Independent Identity}
If $I(S_1;S_2)=0$ and $T=(S_1,S_2)$, then $\operatorname{Red}(S_1,S_2\to T) = 0$.
\end{property}

\begin{remark}
\label{remark:def equal}
From information perspective, there is no difference between ``$T = (S_1, S_2)$'' and ``$H(T|S_1, S_2)=H(S_1,\!S_2|T )\!=\!0$,'' which 
we denote by $T\! \overset{\text{det}}{=}\! (S_1, \!S_2)$ for brevity.
\end{remark}


However, this framework inherently becomes contradictory for three or more source variables, as shown in \cite[Thm.~2]{rauh2014reconsidering}. To set the stage for our results, we briefly recall this finding of the following system and Lemma~\ref{lemma: counter example} proved in Appendix~\ref{app:counter example}.
\paragraph*{System (\(\bar{S}_1,\bar{S}_2,\bar{S}_3,\bar{T}\))}
Let \(x_1\) and \(x_2\) be two independent \(\operatorname{Bernoulli}(1/2)\) variables, and
$x_3 = x_1 \oplus x_2$. Define $\bar{S}_1 = x_1, \bar{S}_2 = x_2, \bar{S}_3 = x_3$, and target $\bar{T} = (x_1, x_2, x_3)$.
\begin{lemma} {\cite{rauh2014reconsidering}}
\label{lemma: counter example}
    For the system $(\bar{S}_1,\bar{S}_2,\bar{S}_3,\bar{T})$, any candidate PID measure $\Pi^{\bar{T}}_{\boldA}:\mathcal{A}(\boldA) \rightarrow \mathbb{R}, \forall \boldA \subseteq \bar{\boldS}$ that satisfies PID Axioms~\ref{PID Axiom: Commutativity}, \ref{pid axiom: Monotonicity}, \ref{pid axiom: Self-redundancy}, Property~\ref{property: Independent Identity}, and PID Axiom~\ref{pid axiom:mutual constrains} for any $|\boldA|\le2$,
    violates PID Axiom~\ref{pid axiom:mutual constrains} for~$|\boldA|=3$, i.e.,
    \begin{align*}
I(T;\bar{\boldS}) < \sum_{\beta \preceq_{\bar{\boldS}} \{\bar{\boldS}\} } \Pi^{\bar{T}}_{\bar{\boldS}}(\beta).
\end{align*}
\end{lemma}

The redundancy lattice of PID was built under the assumption that mutual information could be partitioned into disjoint atoms summing to the total, i.e., WESP. However, Lemma~\ref{lemma: counter example} demonstrates that this assumption is incompatible with synergistic phenomena in multivariate systems.
This phenomena is inherent to multivariate information measures: the sum of synergistic components can exceed the total entropy~\cite{lyu2023system}, which is described rigorously in the following observation. 
\begin{observation}[Synergistic Phenomena \cite{lyu2023system}] 
\label{observation:HleP}
For any system $\boldX = \{X_1,\!\dots\!,X_N\}$, denote $\operatorname{Syn} (\boldX/\hat{X}_{i}\to\! X_i)$ as the information that can only be provided jointly from sources $\boldX/\hat{X}_{i}$ to target $X_i$, the summation of all those information can be greater than the joint entropy of the system, i.e., in general
\begin{align}\label{equation:synEffectLessThanEntropy}
\sum_{i \in \{1,\dots,N\}} \operatorname{Syn} (\boldX /X_i \!\to\! X_i) \nleqslant H(\boldX).
\end{align}
In particular, in the system given in Lemma~\ref{lemma: counter example} we have that the l.h.s of~\eqref{equation:synEffectLessThanEntropy} equals~$3$, whereas the r.h.s equals~$2$.
\end{observation}


In summary, PID’s reliance on the WESP principle (Axiom~\ref{pid axiom:mutual constrains}) is incompatible with the existence of synergy. 
Observation~\ref{observation:HleP} quantifies this incompatibility: the sum of synergistic components can exceed the whole’s entropy. 
Thus, to fix our decomposition, we must relax or modify Axiom~\ref{pid axiom:mutual constrains}. 
In the next section, we propose a new framework addressing this limitation explicitly for three-variable systems.
\section{Three-variable System Information Decomposition}
\label{sec:SID}
In this section we resolve the issue from Section~\ref{sec:PID} for the case of three source variables $S_1,S_2,S_3$, and~$T=(S_1,S_2,S_3)$. 
When $T = (S_1,S_2,S_3)$, the PID of $I(T; S_1,S_2,S_3)$ amounts to a decomposition of the joint entropy $H(S_1,S_2,S_3)$.
We define a new framework in which Axiom~\ref{pid axiom:mutual constrains} is replaced by a summation over a subset of the atoms, in a way which circumvents the overcounting in Observation~\ref{observation:HleP}.
We make use of the following lattice.

\begin{definition}[SID Half Lattice]
\label{definition:SID_half_lattice}
For $\boldS = \{S_1,S_2,S_3\}$, let
\begin{align}\label{equation:lattice}
\mathcal{A}^*(\boldS) =& \{\alpha \in \mathcal{P}_1(\mathcal{P}_1(\boldS)):\nonumber \\\exists& \boldA_k \in \alpha, |\boldA_k|=1,\forall \boldA_i,\boldA_j \in \alpha, \boldA_i \not \subset \boldA_j\}, \\
=&\big\{
\,\{\{1\}\{2\}\{3\}\},\{\{1\}\{2\}\},\{\{1\}\{3\}\},\{\{2\}\{3\}\},\nonumber \\
\{\{1\}\{23&\}\},\{\{2\}\{13\}\}, \{\{3\}\{12\}\},\{\{1\}\},\{\{2\}\},\{\{3\}\}\, \big\}. \nonumber 
\end{align}
where $\mathcal{P}_1(\boldS) $ and $\preceq_\boldS$ are as in Definition~\ref{definition:PID lattice}.
\end{definition} 


For every~$\boldA\subseteq\boldS$ and every~$\alpha\in\cA^*(\boldA)$, our aim is to measure the information contributed by every subset in~$\alpha$ to the whole system $\boldA$, which is not already accounted for by any antichain~$\beta\in\cA^*(\boldA)$ such that~$\beta \preceq_\boldA \alpha$.

\begin{definition} [System Information Decomposition Framework]
\label{sid def:SIDF}
Let~$\boldS$ be a collection of variables.
The set of system information atoms (SI-atoms) is defined as a family of system information functions $\Psi_{\boldA} :\mathcal{A^*}(\boldA) \rightarrow \mathbb{R}$ for all $\boldA \subseteq \boldS$. 
\end{definition}
The SID half lattice can be understood as a refinement of the PID redundancy lattice for three sources (Definition~\ref{definition:PID lattice}), by removing all antichains that do not contain any singleton source (see Figure~\ref{fig:SID_3PID}(B)).
See Appendix~\ref{app:compairson} for further comparison between SID and two-source PID.
\begin{figure}[htbp]
\centering
\fbox{\includegraphics[width=0.97\linewidth]{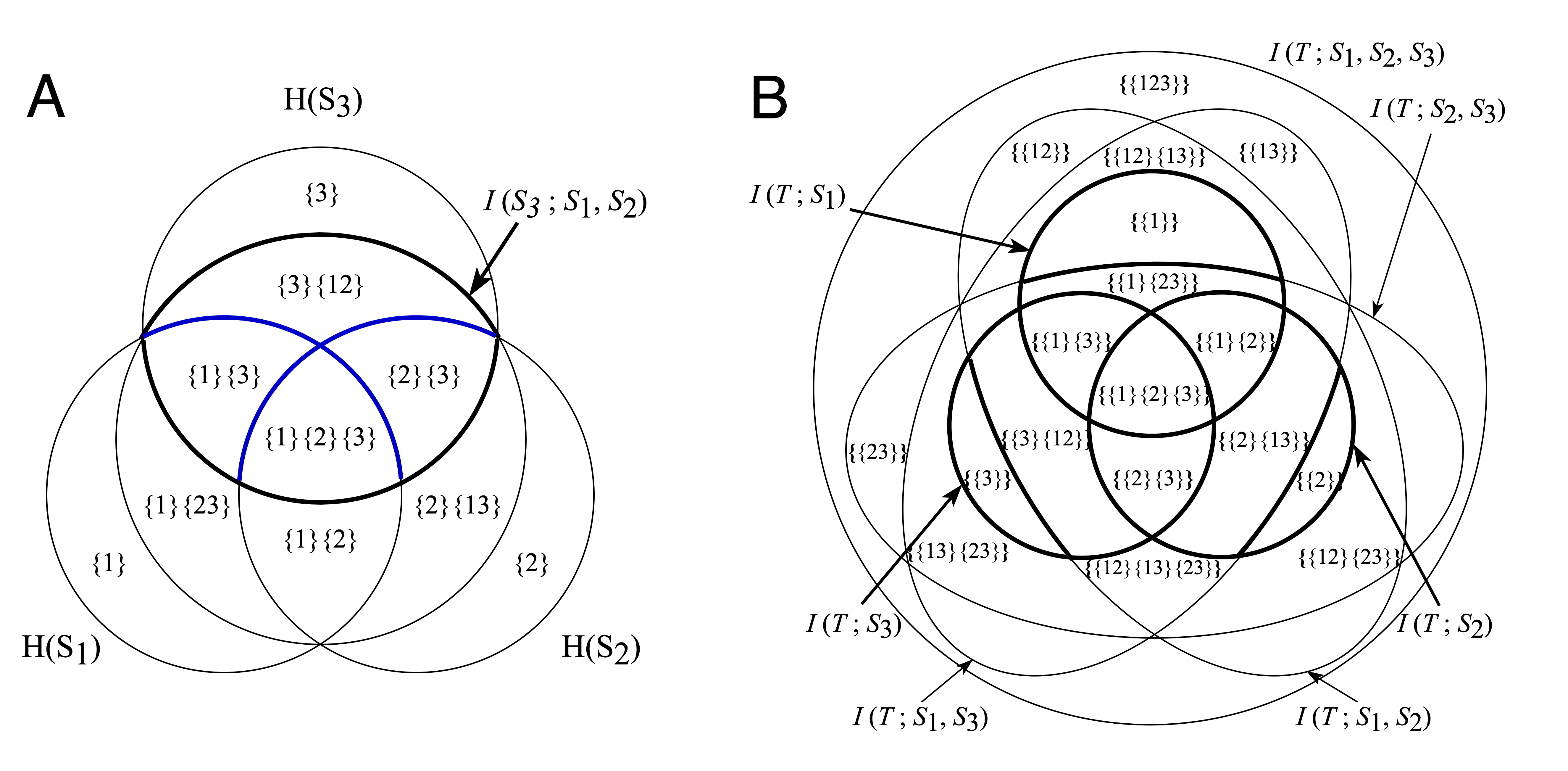 }}
\caption{Comparison between SID and three sources PID. (A) Three-variable SID.
(B) Three-sources PID, where the antichains in bold contain at least one singleton source, whose structure is consistent with SID.}
\label{fig:SID_3PID}
\end{figure}
The original PID Axioms~\ref{PID Axiom: Commutativity}, \ref{pid axiom: Monotonicity}, and~\ref{pid axiom: Self-redundancy}, are required for SID as well. 
PID Axiom~\ref{pid axiom:mutual constrains}, which leads to the inconsistency demonstrated in Lemma~\ref{lemma: counter example}, will be modified shortly.
Similar to PID, we define SID redundant information
as $\operatorname{Red}(S_1,S_2,S_3)=\Psi_{\boldS}(\{\{S_1\}\{S_2\}\{S_3\}\})$, and for all distinct $i,j$ in $\boldS$, let $\operatorname{Red}(S_i,S_j)=\Psi_{\{S_i,S_j\}}(\{\{S_i\}\{S_j\}\})$.




\setcounter{sid axiom}{1}
\begin{sid axiom}[Commutativity]
\label{sid axiom: commutativity}
SID redundant information is invariant under any permutation $\sigma$ of sources, i.e., $    \operatorname{Red}(S_1,S_2,S_3) = \operatorname{Red}(S_{\sigma(1)},S_{\sigma(2)},S_{\sigma(3)})$.
\end{sid axiom}





\begin{sid axiom} [Monotonicity]
\label{sid axiom: Monotonicity}
SID redundant information decreases monotonically as more sources are included, i.e., 
$\operatorname{Red}(S_1,S_2,S_3) \le \min_{i,j\in[3]}\{\operatorname{Red}(S_i,S_j)\}$.
\end{sid axiom}

\begin{sid axiom} [Self-redundancy]
\label{sid axiom: Self-redundancy}
SID redundant information for two variables $S_i,S_j$ equals the mutual information, i.e. $\operatorname{Red}(S_i,S_j) = I(S_i;S_j)$.

\end{sid axiom}




Then, we revisit PID Axiom~\ref{pid axiom:mutual constrains} and aim to propose an alternative axiom. 
In SID, the mutual information between any two variables and the third one can be decomposed similarly to two-source PID. 
That is, for any distinct $i,j,k \in\{1,2,3\}$, $I({S_i,S_j}; S_k)$ splits into four SI-atoms (analogous to~\eqref{equ:Information Atoms' relationship_1}):
\begin{align}
\label{equ:two_one_mutul_decompose}
I(S_i,S_j&;S_k) = \;\Psi_{\boldS}(\{\{i\}\{j\}\{k\}\}) + \Psi_{\boldS}(\{\{i\}\{k\}\}) \nonumber \\
&\phantom{=}+ \Psi_{\boldS}(\{\{j\}\{k\}\}) + \Psi_{\boldS}(\{\{ij\}\{k\}\}),
\end{align}
and the two-variable mutual information $I(S_i; S_k)$ corresponds to two of those atoms (analogous to~\eqref{equ:Information Atoms' relationship_2}):
\begin{align}
\label{equ:one_one_mutul_decompose}
I(S_i;S_k) = \Psi_{\boldS}&(\{\{i\}\{j\}\{k\}\}) + \Psi_{\boldS}(\{\{i\}\{k\}\}).
\end{align}
Recall that we have $H(S_k) = I(S_i,S_j;S_k) + H(S_k|S_i,S_j)$ for any $k\in [3]$, and $\Psi_{\boldS}(\{\{k\}\})$ represents the information provided by~$S_k$ alone, i.e., $ \Psi_{\boldS}(\{\{k\}\})=H(S_k \mid S_i,S_j)$.
Therefore, we have
\begin{align}
\label{eq:single_variable_entropy}
H(S_k) &= \, I(S_i,S_j; S_k) + H(S_k|S_i,S_j)\nonumber \\
&\overset{\eqref{equ:two_one_mutul_decompose}}{=}  \Psi_{\boldS}(\{\{i\}\{j\}\{k\}\})+ \Psi_{\boldS}(\{\{ij\}\{k\}\}) \nonumber \\
+& \Psi_{\boldS}(\{\{j\}\{k\}\}) + \Psi_{\boldS}(\{\{i\}\{k\}\})+\Psi_{\boldS}(\{\{k\}\}) \nonumber \\
&=\sum_{\beta \preceq_{\boldS} \{\{S_k\}\}}\Psi_{\boldS}(\beta).
\end{align}

Similarly, 
for any two variables \(\{S_i,S_k\}\subseteq\boldS\), by combining $H(S_k|S_i)=H(S_k)-I(S_i;S_k)$ 
with~\eqref{equ:one_one_mutul_decompose} and~\eqref{eq:single_variable_entropy}, we have
\begin{align*}
H(S_k|S_i) &= \Psi_{\boldS}(\{\{ij\}\{k\}\}) + \Psi_{\boldS}(\{\{j\}\{k\}\}) + \Psi_{\boldS}(\{\{k\}\}),
\end{align*}
which, combined  with the fact that $H(S_i,S_k)=H(S_i)+H(S_k|S_i)$ and with~\eqref{eq:single_variable_entropy}, shows that the joint entropy of any two variables is the sum of all atoms dominated by that pair:
\begin{align}
\label{eq:two_variables_entropy}
H&(S_i,S_k) = \Psi_{\boldS}(\{\{i\}\{j\}\{k\}\}) + \Psi_{\boldS}(\{\{i\}\{k\}\}) \nonumber \\
&+ \Psi_{\boldS}(\{\{i\}\{j\}\}) + \Psi_{\boldS}(\{\{jk\}\{i\}\})+\Psi_{\boldS}(\{\{i\}\}) \nonumber \\ 
&+\Psi_{\boldS}(\{\{ij\}\{k\}\}) + \Psi_{\boldS}(\{\{j\}\{k\}\}) + \Psi_{\boldS}(\{\{k\}\}) \nonumber \\
&= \Sigma_{\{i,k\}},
\end{align}
where~$\Sigma_{\{i,k\}}$ is the summation of all atoms corresponding to antichains that are dominated either by~$\{\{S_i\}\}$ or by~$\{\{S_k\}\}$.

However, when extending the decomposition to the joint entropy of all three variables, the SID framework deviates from WESP due to the presence of synergy-induced redundancy. 
This discrepancy can be directly demonstrated as follows. 
By combining the fact that $H(S_i,S_j,S_k)=H(S_i,S_k)+H(S_j|S_i,S_k)$ with $\Psi_{\boldS}(\{\{j\}\})=H(S_j|S_i,S_k)$ and~\eqref{eq:two_variables_entropy},
\begin{align}
\label{eq:three_variables_entropy}
H(S_i,S_j,S_k)&=\Sigma_{\{i,k\}} + \Psi_{\boldS}(\{\{j\}\})\nonumber \\
&=\Sigma - \Psi_{\boldS}(\{\{ik\}\{j\}\}),
\end{align}
where~$\Sigma$ is the summation of all~$10$ atoms~$\Psi_{\boldS}(\alpha),\alpha\in\cA^*(\boldS)$.
Thus, unlike PID Axiom~\ref{pid axiom:mutual constrains}, we find that the total entropy is less than the sum of its decomposed parts by exactly $\Psi_{\boldS}(\{\{ij\}\{k\}\})$. 
In other words, WESP does not hold in SID due to this necessary exclusion.

Motivated by~\eqref{eq:single_variable_entropy}, \eqref{eq:two_variables_entropy}, and~\eqref{eq:three_variables_entropy}, 
we propose the following alternative to PID Axiom~\ref{pid axiom:mutual constrains} within the SID framework.

\setcounter{sid axiom}{0}
\begin{sid axiom}
\label{sid axiom:mutual constrains}
For any set of variables $\boldB\subseteq\boldS$ with $|\boldB| \le 2,|\boldS| \le 3$,
the entropy of $\boldB$ is decomposed as~$H(\boldB)=\Sigma_\boldB$, 
and when~$|\boldB|=|\boldS|=3$ we have for all distinct $i,j,k \in[3]$,
\begin{align}
\label{equ:SID Information Atoms}
H(\boldS) = \Sigma - \Psi_{\boldS}(\{\{ij\}\{k\}\}),
\end{align}
where~$\Sigma_B$ and~$\Sigma$ are as defined above. 
\end{sid axiom}


Going back to the example in Lemma~\ref{lemma: counter example} with~$T=(S_1,S_2,S_3)$,
there are three equivalent ways to express~$H(T)$ via Axiom~\ref{sid axiom:mutual constrains}, corresponding to which $\Psi_{\boldS}(\{\{ij\}\{k\}\})$ term is excluded, i.e., for every distinct~$i,j,k\in[3]$ we have
\begin{align*}
H(T)
&=\Psi_{\boldS}(\bigl\{\{i\}\{jk\}\bigl\})+\Psi_{\boldS}(\bigl\{\{j\}\{ik\}\bigl\})=2,
\end{align*} 
where the zero valued SI-atoms are omitted. 





The following lemma shows that only the redundancy atom needs to be defined; the remaining atoms are then uniquely determined via linear constraints implied by Axiom~\ref{sid axiom:mutual constrains}.
A proof is provided in Appendix~\ref{app:UniqueUpToOneAtom} alongside explicit definitions for all SI atoms given~$\operatorname{Red}(S_1,S_2,S_3)$ (see~\eqref{equation:explcitSIatomsGivenRED}).

\begin{lemma}\label{lemma:UniqueUpToOneAtom}
Let \(\boldS = \{S_1, S_2, S_3\}\) be a three-variable system in the SID framework, and
\(\Psi_{123}(\cdot)\) denote its SI-atoms. 
Then, once the value of any one SI-atom is fixed, the values of all remaining SI-atoms are uniquely determined by SID Axiom~\ref{sid axiom:mutual constrains}.
\end{lemma}


Therefore, any definition of $\operatorname{Red}(S_1,S_2,S_3)$ implies unique definitions of all SI atoms that automatically satisfy SID Axiom~\ref{sid axiom:mutual constrains}.
To satisfy Axioms~\ref{sid axiom: commutativity}, \ref{sid axiom: Monotonicity} and~\ref{sid axiom: Self-redundancy}, we adopt an operational definition of redundancy that generalizes the Gács-Körner common information\footnote{The Gács-Körner common information is defined as $\operatorname{CI}(S_1, S_2) \triangleq \max_{Q} H(Q), \text{s.t. } H(Q | S_1) = H(Q | S_2) = 0$.}~\cite{gacs1973common}.

\begin{definition}[Operational Definition of Redundancy] 
\label{definition:red}
For system \( S_1, S_2, S_3 \), the redundant information is defined as $\operatorname{Red}(S_i,S_j)\triangleq I(S_i,S_j)$ for all distinct $i , j\in \{1,2,3\}$, and
\begin{align*}
\operatorname{Red}(S_1,S_2,S_3)&\triangleq \max_{Q}\{H(Q)\mid\!H(Q|S_i)=0, \forall i \!\in \![3]\},
\end{align*} 
where the maximization is taken over all variables \( Q \) defined over the Cartesian product of the alphabets of \( S_1, S_2, S_3 \).
\end{definition}


The following lemma is proved in Appendix~\ref{app: satisfaction}.

\begin{lemma}
\label{lemma: satisfaction}
Definition~\ref{definition:red} satisfies SID~Axioms~\ref{sid axiom: commutativity}, \ref{sid axiom: Monotonicity}, and~\ref{sid axiom: Self-redundancy}.
\end{lemma}


In the next section we detail that new complications arise in systems with four or more variables.
Moreover, we show that they \textit{cannot} be resolved by summing over lattice subsets as done in this section, leaving a challenging open problem.

\section{Limitations of Anti-chains based Information Decomposition}
\label{sec:limitation of ac}
In this section, it is shown that the partial summation approach taken in Section~\ref{sec:SID} is insufficient for three-or-more source variable systems, with a target that is not necessarily equal to the sources. 
\begin{theorem}[Subsystem inconsistency in lattice-based PID]
\label{theorem:no sub set}
For any set of functions $\Pi^{T}_{\boldA}:\mathcal{A}(\boldA) \rightarrow \mathbb{R}, \forall \boldA \subseteq \boldS$, which satisfy PID Axioms~\ref{sid axiom: commutativity}, \ref{pid axiom: Monotonicity}, \ref{pid axiom: Self-redundancy}, and Property~\ref{property: Independent Identity}, there is no way to redefine PID Axiom~\ref{pid axiom:mutual constrains} so that~\eqref{equ:subsystem} (Lemma~\ref{lemma: subsystem consistency}) is satisfied. 
Specifically, there is no fixed subset
\(
\mathcal{O} \subset \mathcal{A}(\boldS)
\)
such that
\begin{align}
\label{equ:subseto}
I(\boldS;T) = \sum_{\beta \in \mathcal{O}} \Pi^{T}_{\boldS}(\beta)
\end{align} 
for all systems with~$|\boldS|=3$ source variables and a target~$T$.
\end{theorem}
We emphasize that Theorem~\ref{theorem:no sub set} 
shows that for any PI-function $\Pi_A$, there is no set $\mathcal{O}$ satisfying~\eqref{equ:subseto} for all systems of three-or-more variables~$\boldS$ and a target~$T$.
In cases where $T=(S_1,S_2,S_3)$ (or equivalently, where~$T\overset{\det}{=}(S_1,S_2,S_3)$) such solution exists as shown in Section~\ref{sec:SID}.
To prove this theorem, we construct the following two systems (Fig.~\ref{fig:System12}) and the subsequent Lemma~\ref{lemma:NoUniversalSubset} is proved in Appendix~\ref{app:NoUniversalSubset}.

\paragraph{System~1 ($\hat{S}_1,\hat{S}_2,\hat{S}_3,\hat{T}$)}
Define six independent bits $x_1,x_2,x_4,x_5,x_7,x_8 \sim \operatorname{Bernoulli}(1/2)$, and three additional
\begin{align*}
    x_3=x_1\oplus x_2,~x_6=x_4\oplus x_5,\mbox{ and } x_9=x_7\oplus x_8.
\end{align*}
Define the system as
$
    \hat{S}_1=(x_1,x_4,x_7),
    \hat{S}_2=(x_2,x_5,x_8),
    \hat{S}_3=(x_3,x_6,x_9), \mbox{ and }
    \hat{T}=(x_1,x_5,x_9).
$

\paragraph{System 2 (\(\bar{S}_1,\bar{S}_2,\bar{S}_3,\bar{T}\))}
Let \(x_1\) and \(x_2\) be two independent \(\operatorname{Bernoulli}(1/2)\) variables, and
$x_3 = x_1 \oplus x_2$. Define $\bar{S}_1 = x_1, \bar{S}_2 = x_2, \bar{S}_3 = x_3$, and set the target as $\bar{T} = (x_1, x_2, x_3)$ (this is the system from Lemma~\ref{lemma: counter example}).
\begin{figure}[htbp]
    \centering
    \fbox{\includegraphics[width=.9\linewidth]{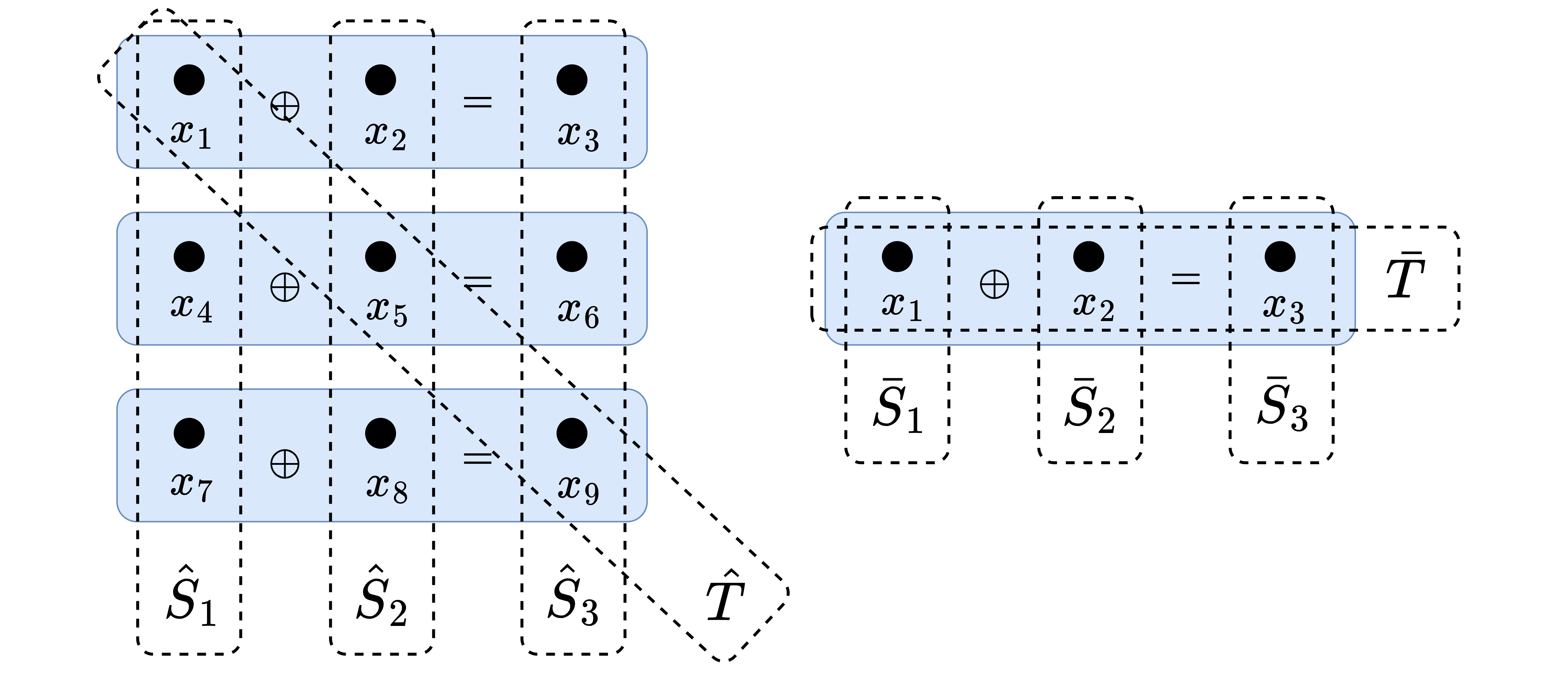}}
    \caption{Construction of $(\hat{S}_1,\hat{S}_2,\hat{S}_3,\hat{T})$ and ($\bar{S}_1,\bar{S}_2,\bar{S}_3,\bar{T}$).}
    \label{fig:System12}
\end{figure}
\begin{lemma}\label{lemma:NoUniversalSubset}
For System 1 $(\hat{\boldS}=\{\hat{S}_1,\hat{S}_2,\hat{S}_3\},\hat{T})$ and System 2 $(\bar{\boldS}=\{\bar{S}_1,\bar{S}_2,\bar{S}_3\},\bar{T})$, (i) there is a bijection $\psi:\mathcal{A}(\hat{\boldS}) \to \mathcal{A}(\bar{\boldS})$ with $\Pi^{\hat{T}}_{\hat{\boldA}}(\hat{\beta})=\Pi^{\bar{T}}_{\psi(\hat{\boldA})}(\psi(\hat{\beta}))$ for all atoms $\hat{\beta} \in \mathcal{A}(\hat{\boldS})$, $\hat{\boldA}\subseteq\hat{\boldS}$, and for any family of $\Pi$ functions satisfying PID Axioms~\ref{sid axiom: commutativity}, \ref{pid axiom: Monotonicity}, \ref{pid axiom: Self-redundancy}, Property~\ref{property: Independent Identity} and~\eqref{equ:subsystem}; and (ii) $I(\hat{\boldS};\hat{T})\ne I(\bar{\boldS};\bar{T})$.
\end{lemma}


Given Lemma~\ref{lemma:NoUniversalSubset}, Theorem~\ref{theorem:no sub set} is immediate: since all atoms are identical, it follows $\sum_{\beta \in \mathcal{O}} \Pi^{\hat{T}}_{\hat{\boldS}}(\hat{\beta})=\sum_{\beta \in \mathcal{O}} \Pi^{\bar{T}}_{\psi(\hat{\boldS})}(\psi(\hat{\beta}))$ for every~$\mathcal{O}$ in the respective lattices, yet~$I(\hat{\boldS};\hat{T})\ne I(\bar{\boldS};\bar{T})$.



This result shows that PID based purely on a fixed collection of antichain-defined atoms will face insurmountable problems when extended beyond three variables. 
This part of the analysis will be discussed in the next section.

\section{Discussion}
\label{sec:discussion}
This work identifies a fundamental inconsistency of the PID framework.
Section~\ref{sec:PID} shows that PID violates the WESP principle in multivariate systems.
To address this, we introduced the SID framework (Section~\ref{sec:SID}), which resolves WESP inconsistencies in three-variable systems by redefining the summation rules for information atoms within a reduced antichain structure. 
However, Section~\ref{sec:limitation of ac} demonstrates via another counterexample that antichain-based PID methods inherently fail to capture higher-order interactions, thus producing unresolvable inconsistency.



\paragraph*{Conceptual Implications of SID}
SID is a target-free commutative framework in which the target is equal to the sources.
While similar to Partial Entropy Decomposition (PED)~\cite{ince2017partial,varley2023partial}, SID uses a reduced antichain lattice (Definition~\ref{definition:SID_half_lattice}) that excludes antichains lacking singleton elements. 
The removed antichains correspond to the irreducible synergistic information atoms that cannot be provided by any single source alone; these are unnecessary when the target is a subset of the sources. 
This insight is inspired by the decomposition of the joint entropy of a two-variable system using Shannon's laws, i.e. $H(X,Y)=H(X|Y)+H(Y)$, where each component of the decomposition comes directly from at least one of the source variables. 
This insight can be easily extended through the chain rule~\cite[Ch. 2]{yeung2012first} to multivariable systems.  
SID Axiom~\ref{sid axiom:mutual constrains} is motivated by these observations, and furthermore, it reveals that synergistic information is inherently symmetric and higher-order, leading to multivariate information interactions that fundamentally differ from classical set theory.

\paragraph*{Beyond Three Variables – Limitations of Antichains}
Section~\ref{sec:limitation of ac} illustrates scenarios where two distinct systems share identical PID decompositions yet differ in total joint information. 
This discrepancy arises because synergy is fundamentally a \textit{holistic} (or \textit{collective}) higher-order symmetric relationship among variables: an $n$-order synergy (among $n$ variables) contributes its information collectively to the system as a whole, appearing~$n-1$ times in the summation rule for joint entropy (SID Axiom~\ref{sid axiom:mutual constrains}). 
PID frameworks fail to capture this holistic nature of synergy, as they focus solely on  mutual information (i.e., $I(\boldS;T)$), and isolate information to single antichain atoms, while ignoring the collective synergistic structure.
Thus, we suggest that future multivariate decomposition frameworks must:
(1) Decompose total system entropy rather than partial/mutual information.
(2) Capture higher-order synergistic structures across variable subsets.
(3) Move beyond traditional set-theoretic assumptions (e.g., WESP) to accommodate information’s unique properties.

\newpage
\IEEEtriggeratref{11}
\bibliographystyle{unsrt}
\bibliography{references}




\newpage
\appendix
\subsection{Comparison between SID and two sources PID}
\label{app:compairson}

\begin{figure}[htbp]
\centering
\fbox{\includegraphics[width=0.97\linewidth]{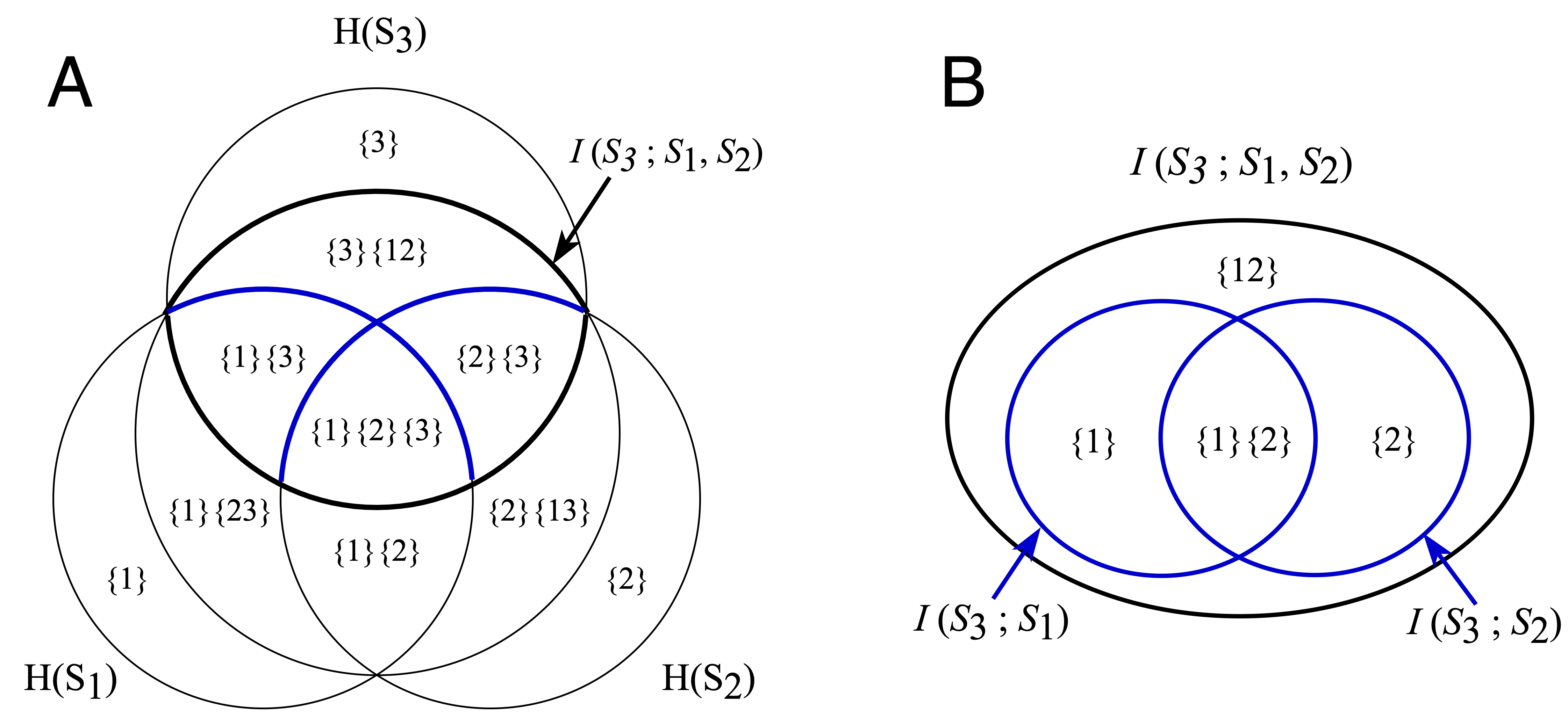 }}
\caption{Comparison between SID and two sources PID.}
\label{fig:SID_2PID}
\end{figure}

SID extends the scope of 2-source PID from mutual information $I(\boldS/S_i;S_i)$ to the joint entropy $H(\boldS)$ of the system. 
In SID (target-free), each SI-atom represents information that a certain combination of variables provides redundantly to the system as a whole.
For instance, in Fig.~\ref{fig:SID_2PID}(A), the SI-atom $\Psi_{123}(\{\{3\}\{12\}\})$ represents information in $S_3$ that is also contributed synergistically by $S_1$ and $S_2$. This directly corresponds to the PI-atom $\Pi_{12}^3(\{\{12\}\})$ in the PID view (Fig.~\ref{fig:SID_2PID}(B)), where we have a target $T=S_3$ and sources $S_1,S_2$.

\subsection{Proofs of Main Results}
To prove the lemmas in the paper, we first need the following corollary from Lemma~\ref{lemma: subsystem consistency}.
\begin{corollary}
\label{corollary: two result}
For the system $(S_1, S_2,S_3,T)$ and its sub-system $(S_1,\!S_2,\!T)$, $(S_1,\!T)$, the decomposed PI-atoms from different sub-systems have the following relationship:
\begin{align}
\label{equ:cross scale_0}
\Pi^{T}_{1}(\!\bigl\{\!\{1\}\!\bigl\}\!) = \Pi^{T}_{12}(\!\bigl\{\!\{1\}\{2\}\!\bigl\}\!) +\Pi^{T}_{12}(\!\bigl\{\!\{1\}\!\bigl\}\!),
\end{align}  
similarly, for the system $(S_1,S_2,S_3,T)$ and $(S_1,S_2,T)$,
\begin{align}
\label{equ:cross scale}
\Pi^{T}_{12}(\!\bigl\{\!\{1\}\{2\}\!\bigl\}\!) &= \Pi^{T}_{123}(\!\bigl\{\!\{1\}\!\{2\}\!\{3\}\!\bigl\}\!) +\Pi^{T}_{123}(\!\bigl\{\!\{1\}\!\{2\}\!\bigl\}\!),\\
\Pi^{T}_{12}(\!\bigl\{\!\{1\}\!\bigl\}\!) &= \Pi^{T}_{123}(\!\bigl\{\!\{1\}\{3\}\!\bigl\}\!)+\Pi^{T}_{123}(\!\bigl\{\!\{1\}\{23\}\!\bigl\}\!)\nonumber \\&+\Pi^{T}_{123}(\!\bigl\{\!\{1\}\!\bigl\}\!).\label{equ:cross scale_2}
\end{align}
\end{corollary}
\begin{proof}
For the system $(S_1, S_2, T)$ and $(S_1, T)$, according to Lemma~\ref{lemma: subsystem consistency}, let $\boldA=\{S_1,S_2\}$, $\boldB=\{S_1\}$, and $\boldC=\{S_1\}$, then we have
\begin{align}\label{equ:decompose_of_0}
\Pi^{T}_{1}(\!\bigl\{\!\{1\}\!\bigl\}\!) = \Pi^{T}_{12}(\!\bigl\{\!\{1\}\{2\}\!\bigl\}\!) +\Pi^{T}_{12}(\!\bigl\{\!\{1\}\!\bigl\}\!).
\end{align}  
Similarly, for the system $(S_1, S_2, T)$ and $(S_2, T)$, we have 
\begin{align}
\Pi^{T}_{2}(\!\bigl\{\!\{2\}\!\bigl\}\!) = \Pi^{T}_{12}(\!\bigl\{\!\{1\}\{2\}\!\bigl\}\!) +\Pi^{T}_{12}(\!\bigl\{\!\{2\}\!\bigl\}\!),\nonumber
\end{align}  
where the information atoms contained in both $\Pi^{T}_{1}(\!\bigl\{\!\{1\}\!\bigl\}\!)$ and $\Pi^{T}_{2}(\!\bigl\{\!\{2\}\!\bigl\}\!)$ is $\Pi^{T}_{12}(\!\bigl\{\!\{1\}\{2\}\!\bigl\}\!)$.

Then, following the same approach, we focus on the system $(S_1,S_2,S_3,T)$ and $(S_1, T)$, i.e., we let $\boldA=\{S_1,S_2,S_3\}$, $\boldB=\{S_1\}$, and $\boldC=\{S_1\}$.
Then, by~Lemma~\ref{lemma: subsystem consistency} we have
\begin{align}\label{equ:decompose_of_1}
\Pi^{T}_{1}&(\!\bigl\{\!\{1\}\!\bigl\}\!) =\Pi^{T}_{123}(\!\bigl\{\!\{1\}\{2\}\{3\}\!\bigl\}\!) +\Pi^{T}_{123}(\!\bigl\{\!\{1\}\{2\}\!\bigl\}\!)\nonumber\\
&+\Pi^{T}_{123}(\!\bigl\{\!\{1\}\{3\}\!\bigl\}\!)+\Pi^{T}_{123}(\!\bigl\{\!\{1\}\{23\}\!\bigl\}\!)+\Pi^{T}_{123}(\!\bigl\{\!\{1\}\!\bigl\}\!).
\end{align}  
Similarly, for the system $(S_1, S_2,S_3,T)$ and $(S_2, T)$, we have
\begin{align}
\Pi^{T}_{2}&(\!\bigl\{\!\{2\}\!\bigl\}\!) =\Pi^{T}_{123}(\!\bigl\{\!\{1\}\{2\}\{3\}\!\bigl\}\!) +\Pi^{T}_{123}(\!\bigl\{\!\{1\}\{2\}\!\bigl\}\!)\nonumber\\
&+\Pi^{T}_{123}(\!\bigl\{\!\{2\}\{3\}\!\bigl\}\!)+\Pi^{T}_{123}(\!\bigl\{\!\{2\}\{13\}\!\bigl\}\!)+\Pi^{T}_{123}(\!\bigl\{\!\{2\}\!\bigl\}\!),\nonumber
\end{align}  
where the information atoms contained in both $\Pi^{T}_{1}(\!\bigl\{\!\{1\}\!\bigl\}\!)$ and $\Pi^{T}_{2}(\!\bigl\{\!\{2\}\!\bigl\}\!)$ are $\Pi^{T}_{123}(\!\bigl\{\!\{1\}\{2\}\{3\}\!\bigl\}\!) $ and $\Pi^{T}_{123}(\!\bigl\{\!\{1\}\{2\}\!\bigl\}\!)$.
Hence, we have
\begin{align}
\Pi^{T}_{12}(\!\bigl\{\!\{1\}\{2\}\!\bigl\}\!) &= \Pi^{T}_{123}(\!\bigl\{\!\{1\}\{2\}\{3\}\!\bigl\}\!) +\Pi^{T}_{123}(\!\bigl\{\!\{1\}\{2\}\!\bigl\}\!),\nonumber
\end{align}  
where $\{\Pi^{T}_{12}(\!\bigl\{\!\{1\}\{2\}\!\bigl\}\!)\}$ and $\{\Pi^{T}_{123}(\!\bigl\{\!\{1\}\{2\}\{3\}\!\bigl\}\!)$, $ \Pi^{T}_{123}(\!\bigl\{\!\{1\}\{2\}\!\bigl\}\!)\}$ are the only atom(s) that are contained in both $I(S_1,T)$ (i.e., $\Pi^{T}_{1}(\!\bigl\{\!\{1\}\!\bigl\}\!)$) and $I(S_2,T)$ (i.e., $\Pi^{T}_{2}(\!\bigl\{\!\{2\}\!\bigl\}\!)$) from the decompositions under the scope of $(S_1, S_2,T)$ and $(S_1, S_2,S_3,T)$.
Therefore,~\eqref{equ:cross scale} is proved. 

Then, by~\eqref{equ:cross scale}, \eqref{equ:decompose_of_0}, and \eqref{equ:decompose_of_1}, we have
\begin{align*}
\Pi^{T}_{12}(\!\bigl\{\!\{1\}\!\bigl\}\!) \!=\! \Pi^{T}_{123}(\!\bigl\{\!\{1\}\!\{3\}\!\bigl\}\!)\!+\!\Pi^{T}_{123}(\!\bigl\{\!\{1\}\!\{23\}\!\bigl\}\!)\!+\!\Pi^{T}_{123}(\!\bigl\{\!\{1\}\!\bigl\}\!),
\end{align*}
which is~\eqref{equ:cross scale_2}.
\end{proof}
Using Corollary~\ref{corollary: two result}, we prove Lemma~\ref{lemma: counter example}, \ref{lemma:UniqueUpToOneAtom}, \ref{lemma: satisfaction}, and~\ref{lemma:NoUniversalSubset} sequentially.
\subsubsection{Proof of Lemma~\ref{lemma: counter example}}
\label{app:counter example}
\begin{proof}
In $(\bar{S}_1,\bar{S}_2,\bar{S}_3,\bar{T})$, let $\bar{S}_1$ and $\bar{S}_2$ be two independent~$\text{Bernoulli}(1/2)$ variables, let $\bar{S}_3 = \bar{S}_1 \oplus \bar{S}_2$, and let~$\bar{T}=(\bar{S}_1,\bar{S}_2,\bar{S}_3)$. Therefore, we have
\begin{align}
\label{equ:I(T;S_1,S_2,S_3)}
I(\bar{T};\bar{S}_1,\bar{S}_2,\bar{S}_3)=2.
\end{align}

Our subsequent proof idea is to use Property~\ref{property: Independent Identity} to obtain the values of all PI-atoms in any system with two sources and the target variable (i.e. $(\bar{S}_1,\bar{S}_2,\bar{T}),(\bar{S}_1,\bar{S}_3,\bar{T})$ and $(\bar{S}_2,\bar{S}_3,\bar{T})$) and then show that their sum will be greater than the joint mutual information of the system $(\bar{S}_1,\bar{S}_2,\bar{S}_3,\bar{T})$. For simplicity, throughout the following proof, we adopt the convention that all statements are considered for distinct~$i,j,k \in \{1,2,3\}$. 

Firstly, by Property~\ref{property: Independent Identity} (Independent Identity and Remark~\ref{remark:def equal}), and since $\bar{T} =(\bar{S}_1,\bar{S}_2,\bar{S}_3)\overset{\text{det}}{=}(\bar{S}_i,\bar{S}_j)$ we have
\begin{align}
\label{equ:all_three_is_zero}
     \Pi^{\bar{T}}_{ij}(\bigl\{\{i\}\{j\}\bigl\}) = 0.
\end{align}
Considering that
    $\Pi^{\bar{T}}_{ij}(\bigl\{\{i\}\{j\}\bigl\})= \Pi^{\bar{T}}_{123}(\bigl\{\{1\}\{2\}\{3\}\bigl\}) + \Pi^{\bar{T}}_{123}(\bigl\{\{i\}\{j\}\bigl\})$,
which is identical to~\eqref{equ:cross scale}, and by Axiom~\ref{pid axiom: Monotonicity} (Monotonicity) and Lemma~\ref{Lemma: Nonnegativity} (Nonnegativity) we have 
\begin{align}
\label{equ:pi12=0}
    \Pi^{\bar{T}}_{123}(\bigl\{\{1\}\{2\}\{3\}\bigl\}) = \Pi^{\bar{T}}_{123}(\bigl\{\{i\}\{j\}\bigl\}) =  0.
\end{align}
Similarly, \eqref{equ:Information Atoms' relationship_2} implies that $I(\bar{T};\bar{S}_i)= \Pi^{\bar{T}}_{ij}(\bigl\{\{i\}\{j\}\bigl\}) + \Pi^{\bar{T}}_{ij}(\bigl\{\{i\}\bigl\})$, and since $ I(\bar{T};\bar{S}_i) =1$ and due to~\eqref{equ:all_three_is_zero}, it follows that $\Pi^{\bar{T}}_{ij}(\bigl\{\{i\}\bigl\})=1$, 
which by Corollary~\ref{corollary: two result} (specifically~\eqref{equ:cross scale_2}), equals
\begin{align}
\label{equ:pi-red}
    \Pi^{\bar{T}}_{123}(\bigl\{\{i\}\bigl\}) + \Pi^{\bar{T}}_{123}(\bigl\{\{i\}\{jk\}\bigl\}) + \Pi^{\bar{T}}_{123}(\bigl\{\{i\}\{k\}\bigl\}).
\end{align}
Then, by \eqref{equ:pi12=0} and \eqref{equ:pi-red}, we have
\begin{align}
\label{equ:sun=1}
\Pi^{\bar{T}}_{123}(\bigl\{\{i\}\bigl\})+ \Pi^{\bar{T}}_{123}(\bigl\{\{i\}\{jk\}\bigl\}) = 1,
\end{align}
and hence,
\begin{align*}
I(\bar{T};\bar{S}_1,\bar{S}_2,\bar{S}_3) &\ge \Pi^{\bar{T}}_{123}(\bigl\{\{1\}\bigl\}) + \Pi^{\bar{T}}_{123}(\bigl\{\{1\}\{23\}\bigl\})  \\&+ \Pi^{\bar{T}}_{123}(\bigl\{\{2\}\bigl\})\nonumber+ \Pi^{\bar{T}}_{123}(\bigl\{\{2\}\{13\}\bigl\}) \\&+ \Pi^{\bar{T}}_{123}(\bigl\{\{3\}\bigl\}) + \Pi^{\bar{T}}_{123}(\bigl\{\{3\}\{12\}\bigl\}) =3,
\end{align*}
which contradicts~\eqref{equ:I(T;S_1,S_2,S_3)}.
\end{proof}

\subsubsection{Proof of Lemma~\ref{lemma:UniqueUpToOneAtom}}
\label{app:UniqueUpToOneAtom}
\begin{proof}
We consider the linear constraints relating to the following ten unknowns (the ten SI-atoms of a three-variable system). Define the following vector of atoms:
\begin{align*}
X = \Bigl[
&\,\Psi_{123}(\{\{1\}\{2\}\{3\}\}),\\
&\Psi_{123}(\{\{1\}\{2\}\}),\Psi_{123}(\{\{1\}\{3\}\}),\Psi_{123}(\{\{2\}\{3\}\}),\\
&\Psi_{123}(\{\{1\}\{23\}\}),\Psi_{123}(\{\{2\}\{13\}\}),\Psi_{123}(\{\{3\}\{12\}\}),\\
&\Psi_{123}(\{\{1\}\}),\,\Psi_{123}(\{\{2\}\}),\,\Psi_{123}(\{\{3\}\}) \Bigr]^T,
\end{align*}
and the following vector of entropies:
\begin{align*}
Y \;=\; \Bigl[
&\;H(S_1),\;H(S_2),\;H(S_3),\\
&\,H(S_1,S_2),\;H(S_1,S_3),\;H(S_2,S_3),\\
&\,H(S_1,S_2,S_3),\;H(S_1,S_2,S_3),\;H(S_1,S_2,S_3)
\Bigr]^T.
\end{align*}
Then, the nine constraints which arise from SID Axiom~\ref{sid axiom:mutual constrains}, along with the conditions from SID Axiom~\ref{sid axiom:mutual constrains} are as follows.
\begin{align*}
\begin{bmatrix}
1 & 1 & 1 & 0 & 1 & 0 & 0 & 1 & 0 & 0\\
1 & 1 & 0 & 1 & 0 & 1 & 0 & 0 & 1 & 0\\
1 & 0 & 1 & 1 & 0 & 0 & 1 & 0 & 0 & 1\\
1 & 1 & 1 & 1 & 1 & 1 & 0 & 1 & 1 & 0\\
1 & 1 & 1 & 1 & 1 & 0 & 1 & 1 & 0 & 1\\
1 & 1 & 1 & 1 & 0 & 1 & 1 & 0 & 1 & 1\\
1 & 1 & 1 & 1 & 1 & 1 & 0 & 1 & 1 & 1\\
1 & 1 & 1 & 1 & 1 & 0 & 1 & 1 & 1 & 1\\
1 & 1 & 1 & 1 & 0 & 1 & 1 & 1 & 1 & 1
\end{bmatrix}
X \;=\; Y.
\end{align*}

Solving the system provides the following definition of all SI atoms given~$\operatorname{Red}(S_1,S_2,S_3)$:
\begin{align}\label{equation:explcitSIatomsGivenRED}
    \Psi_{123}(\{1\}\{2\}\{3\})&\triangleq\operatorname{Red}(S_{1},S_2,S_{3})\nonumber\\
    \Psi_{123}(\big\{\{ij\}\big\})&=H(S_i)+H(S_j)\nonumber\\
    &-H(S_i,S_j)-\operatorname{Red}(S_{1},S_2,S_{3})\nonumber\\
    \Psi_{123}(\big\{\{i\}\{jk\}\big\})&=-H(S_1)-H(S_2)-H(S_3)\nonumber\\
    &+H(S_1,S_2)+H(S_1,S_3)+H(S_2,S_3)\nonumber\\
    &-H(S_1,S_2,S_3)+\operatorname{Red}(S_{1},S_2,S_{3}) \nonumber \\
    \Psi_{123}(\big\{\{i\}\big\})&=H(S_1,S_2,S_3)-H(S_j,S_k) 
\end{align}
for all~$i,j,k$ such that \(\{i,j,k\} = \{1,2,3\}\).
\end{proof}



\subsubsection{Proof of Lemma~\ref{lemma: satisfaction}}
\label{app: satisfaction}
\begin{proof}~
SID Axiom~\ref{sid axiom: commutativity} (Commutativity) is clearly satisfied by Definition~\ref{definition:red}, since the condition is symmetric with respect to the input variables; SID Axiom~\ref{sid axiom: Self-redundancy} (Self-redundancy) is also satisfied by the definition. 
SID Axiom~\ref{sid axiom: Monotonicity} (Monotonicity) follows from Definition~\ref{definition:red} since adding a new variable imposes additional constraints on the maximization:
\begin{align*}
    \operatorname{Red}&(S_1,S_2,S_3) = \max_{Q}\{H(Q) : H(Q \mid S_i) = 0, \forall i \in [3] \} \\ 
    &\leq \max_{Q} \{H(Q) : H(Q \mid S_i) = 0, H(Q \mid S_j) = 0 \} \\ 
    &= \operatorname{CI}(S_i, S_j), 
\end{align*}
for every distinct~$i$ and~$j$ in~$\{1,2,3\}$, where the last equality follows from the definition $\operatorname{CI}(S_1, S_2) \triangleq \max_{Q} H(Q), \text{s.t. } H(Q | S_1) = H(Q | S_2) = 0$~\cite{gacs1973common}. Moreover, since \(\operatorname{CI}(S_i, S_j) \leq I(S_i; S_j)\) \cite{gacs1973common},
it follows that
\begin{align*}
    \operatorname{Red}(S_1, S_2, S_3) &\leq I(S_i; S_j), 
\end{align*}
for every distinct~$i$,~$j$ in~$\{1,2,3\}$, hence SID Axiom~\ref{sid axiom: Monotonicity} follows.
\end{proof}

\subsubsection{Proof of Lemma~\ref{lemma:NoUniversalSubset}}
\label{app:NoUniversalSubset}

\begin{proof}
Recall that System~1 contains six independent $\operatorname{Bernoulli}(1/2)$ variables $x_1,x_2,x_4,x_5,x_7,x_8$, three additional bits 
\begin{align}\label{equation:S1XOR}
    x_3=x_1\oplus x_2,~x_6=x_4\oplus x_5,~x_9=x_7\oplus x_8,
\end{align}
the source variables are $$\hat{S}_1=(x_1,x_4,x_7), \hat{S}_2=(x_2,x_5,x_8), \hat{S}_3=(x_3,x_6,x_9),$$the target is $\hat{T}=(x_1,x_5,x_9),$ and
\begin{align*}
    I(\hat{T}; \hat{S}_1, \hat{S}_2, \hat{S}_3) &= H(\hat{T}) =H(X_1,X_5,X_9)= 3.
\end{align*}

Intuitively, observe that each one of $x_1, x_5, x_9$ can be predicted by the other two sources in its respective XOR equation~\eqref{equation:S1XOR}. 
Therefore, given $\hat{S}_2$ and $\hat{S}_3$ one can determine $x_1$. Thus, the information about $x_1$ is uniquely shared by $\hat{S}_2$ and $\hat{S}_3$ in a synergistic way. This corresponds to the PID atom $\Pi^{\hat{T}}_{123}(\{\{1\}\{23\}\}) = H(x_1) = 1$.
Similarly, $x_5$ = $x_4 \oplus x_6$ implies that $\hat{S}_1$ and $\hat{S}_3$ together determine $x_5$, yielding $\Pi^{\hat{T}}_{123}(\{\{2\}\{13\}\}) = H(x_5) = 1$, and $x_9$ = $x_7 \oplus x_8$ implies that $\hat{S}_1 $ and $\hat{S}_2$ together determine $x_9$, yielding $\Pi^{\hat{T}}_{123}(\{\{3\}\{12\}\}) = H(x_9) = 1$. 
Then, all remaining atoms are zero in this system.
A similar intuition System~2 is given below.



To justify the above claims regarding System~1 in formal terms, we analyze the system by considering the following three sub-target variables separated from target $\hat{T}$:
\[
\hat{T}_1 = x_1, \quad \hat{T}_2 = x_5, \quad \hat{T}_3 = x_9,
\]
where \( \hat{T} = (\hat{T}_1, \hat{T}_2, \hat{T}_3) \) and the three sub-targets are mutually independent.
 
We operate in three steps, where in the first we identify the zero PI-atoms, in the second we identify the nonzero ones, and in the third we combine the conclusions.

\textbf{Step 1:} Establishing zero PI-atoms.
By PID Axiom~\ref{pid axiom: Self-redundancy} we have~$\operatorname{Red}(\hat{S}_i \to \hat{T}) = I(\hat{S}_i;\hat{T})$, and since
\[
I(\hat{S}_2; \hat{T}_1) = I(\hat{S}_3; \hat{T}_1) = 0,
\]
it follows that
\begin{align}
\label{equ:I(S_i;T_1)=0}
\Pi^{\hat{T}_1}_{2}(\{\{2\}\}) = \Pi^{\hat{T}_1}_{3}(\{\{3\}\}) = 0    
\end{align}

Next, by 
PID Axiom~\ref{pid axiom: Monotonicity} (monotonicity), we have $\Pi^{\hat{T}_1}_{ij}(\{\{i\}\{j\}\}) \le \Pi^{\hat{T}_1}_{i}(\{\{i\}\})$, where $\Pi^{\hat{T}_1}_{i}(\{\{i\}\})=0$ for every~$i\ne 1$ by~\eqref{equ:I(S_i;T_1)=0}.
Therefore, since Lemma~\ref{Lemma: Nonnegativity} (nonnegativity) implies that $\Pi^{\hat{T}_1}_{ij}(\{\{i\}\{j\}\}) \ge 0$, we have
\begin{align}
     \label{equ:pi12=0_1}
    \Pi^{\hat{T}_1}_{ij}(\{\{i\}\{j\}\}) &= 0, \quad \forall i \neq j \in \{1,2,3\}.
\end{align}

Similarly, applying the consistency from Lemma~\ref{lemma: subsystem consistency}, i.e. \eqref{equ:cross scale}, to any system $(\hat{S}_i,\hat{S}_j,\hat{T}_1)$ and to $(\hat{S}_1,\hat{S}_2,\hat{S}_3,\hat{T}_1)$, we have 
\begin{align}
    \Pi^{\hat{T}_1}_{ij}(\{\{i\}\{j\}\}) &= \Pi^{\hat{T}_1}_{123}(\{\{1\}\{2\}\{3\}\}) + \Pi^{\hat{T}_1}_{123}(\{\{i\}\{j\}\}),  
    \label{equ:cross_scale_2}
\end{align}
and by Axiom~\ref{pid axiom: Monotonicity} (monotonicity) and Lemma~\ref{Lemma: Nonnegativity} (nonnegativity) we obtain:
\begin{align}
    0\le \Pi^{\hat{T}_1}_{123}(\{\{1\}\{2\}\{3\}\}) \le \Pi^{\hat{T}_1}_{ij}(\{\{i\}\{j\}\}) \overset{\eqref{equ:pi12=0_1}}{=} 0, \quad \forall i \neq j, \nonumber
\end{align}
which implies that 
\begin{align}\label{equation:PIzero123}
    \Pi^{\hat{T}_1}_{123}(\{\{1\}\{2\}\{3\}\})=0.
\end{align}
Then, by~\eqref{equ:cross_scale_2} and~\eqref{equ:pi12=0_1}, we have 
\begin{align}
    \label{equ:all_zero}
    \Pi^{\hat{T}_1}_{123}(\{\{i\}\{j\}\}) = 0, \quad \forall i \neq j.
\end{align}



Finally, observe that $H(\hat{T}_1,\hat{S}_1|\hat{S}_2, \hat{S}_3)=0$, i.e., the entire system entropy is provided by~$\hat{S}_2,\hat{S}_3$.
Therefore, all PID atoms that do not include either $\{\hat{S}_2\}$, or
$\{\hat{S}_3\}$ or $\{\hat{S}_2,\hat{S}_3\}$
are zero, i.e.,
\begin{align}
\label{equ:other_zero}
\Pi^{\hat{T}_1}_{123}(\{\{1\}\}) =   \Pi^{\hat{T}_1}_{123}(\{\{12\}\})=\Pi^{\hat{T}_1}_{123}(\{\{13\}\})\nonumber \\=\Pi^{\hat{T}_1}_{123}(\{\{123\}\})=\Pi^{\hat{T}_1}_{123}(\{\{12\}\{13\}\})=0.
\end{align}

\textbf{Step 2:} Determining non-zero PI-atoms.
First, we have $I(\hat{T}_1; \hat{S}_1) = 1$, and then, we use PID Axiom~\ref{pid axiom: Self-redundancy} ($I(\hat{T}_1; \hat{S}_1)=\Pi^{\hat{T}_1}_{1}(\{\{1\}\})$) and Lemma~\ref{lemma: subsystem consistency}, i.e.,~\eqref{equ:subsystem} to get
\begin{align}
    &\Pi^{\hat{T}_1}_{1}\!(\{\!\{1\}\!\})=\Pi^{\hat{T}_1}_{123}(\{\{1\}\{2\}\{3\}\}) \!+\! \Pi^{\hat{T}_1}_{123}(\{\{1\}\{23\}\}) \\ \nonumber
    &+ \Pi^{\hat{T}_1}_{123}(\{\{1\}\{3\}\}) + \Pi^{\hat{T}_1}_{123}(\{\{1\}\{2\}\}) +\Pi^{\hat{T}_1}_{123}(\{\{1\}\}) = 1.
\end{align}
Combining this with \eqref{equation:PIzero123}, \eqref{equ:all_zero}, and \eqref{equ:other_zero}, we obtain:
\[
\Pi^{\hat{T}_1}_{123}(\{\{1\}\{23\}\}) = 1.
\]
The same reasoning applies symmetrically to \( (\hat{S}_1, \hat{S}_2, \hat{S}_3, \hat{T}_2) \) and \( (\hat{S}_1, \hat{S}_2, \hat{S}_3, \hat{T}_3) \), yielding $\Pi^{\hat{T}_2}_{123}(\{\{2\}\{13\}\}) = 1$ 
and $\Pi^{\hat{T}_3}_{123}(\{\{3\}\{12\}\}) = 1$.

\textbf{Step 3:} Final conclusion.
Since \( \hat{T} = (\hat{T}_1, \hat{T}_2, \hat{T}_3) \) and the three sub-targets are independent, the final conclusion for the system \( (\hat{S}_1, \hat{S}_2, \hat{S}_3, \hat{T}) \) is:
\begin{align*}
\Pi^{\hat{T}}_{123}(\{\{i\}\{jk\}\}) = 1, \quad \text{for all distinct}~i,j,k \in [3].
\end{align*}


We now turn to discuss System 2 (\(\bar{S}_1,\bar{S}_2,\bar{S}_3,\bar{T}\)). 
Recall that System~2 contains two independent $\operatorname{Bernoulli}(1/2)$ variables~$x_1,x_2$ and their exclusive or~$x_3=x_1\oplus x_2$.
%
Then, $\bar{S}_1 = x_1$,  
$\bar{S}_2 = x_2$, 
$\bar{S}_3 = x_3$,
$\bar{T}   = (x_1, x_2, x_3),$ 
and 
\begin{align*}
    I(\bar{T}; \bar{S}_1, \bar{S}_2, \bar{S}_3) &= H(\bar{T}) =H(x_1,x_2,x_3)= 2.
\end{align*}

Firstly, by Property~\ref{property: Independent Identity} (Independent Identity and Remark~\ref{remark:def equal}), and since $\bar{T} =(\bar{S}_1,\bar{S}_2,\bar{S}_3)\overset{\text{det}}{=}(\bar{S}_i,\bar{S}_j)$ for every distinct $i,j\in[3]$, we have
\begin{align}
\label{equ:all_two_is_zero}
     \Pi^{\bar{T}}_{ij}(\bigl\{\{i\}\{j\}\bigl\}) = 0.
\end{align}
Considering that
    $\Pi^{\bar{T}}_{ij}(\bigl\{\{i\}\{j\}\bigl\})= \Pi^{\bar{T}}_{123}(\bigl\{\{1\}\{2\}\{3\}\bigl\}) + \Pi^{\bar{T}}_{123}(\bigl\{\{i\}\{j\}\bigl\})$,
which is identical to~\eqref{equ:cross scale}, and by Axiom~\ref{pid axiom: Monotonicity} (Monotonicity) and Lemma~\ref{Lemma: Nonnegativity} (nonnegativity) we have 
\begin{align}
\label{equ:pi12=0_2}
    \Pi^{\bar{T}}_{123}(\bigl\{\{1\}\{2\}\{3\}\bigl\}) = \Pi^{\bar{T}}_{123}(\bigl\{\{i\}\{j\}\bigl\}) =  0
\end{align}
Similarly, \eqref{equ:Information Atoms' relationship_2} implies that $I(\bar{T};\bar{S}_i)= \Pi^{\bar{T}}_{ij}(\bigl\{\{i\}\{j\}\bigl\}) + \Pi^{\bar{T}}_{ij}(\bigl\{\{i\}\bigl\})$, and moreover since $ I(\bar{T};\bar{S}_i) =1$ and~\eqref{equ:all_two_is_zero}, it follows that $\Pi^{\bar{T}}_{ij}(\bigl\{\{i\}\bigl\})=1$, 
which by Corollary~\ref{corollary: two result}, \eqref{equ:cross scale_2}, equals
\begin{align}
\label{equ:pi-red_2}
    \Pi^{\bar{T}}_{123}(\bigl\{\{i\}\bigl\}) + \Pi^{\bar{T}}_{123}(\bigl\{\{i\}\{jk\}\bigl\}) + \Pi^{\bar{T}}_{123}(\bigl\{\{i\}\{k\}\bigl\}).
\end{align}
Then, by \eqref{equ:pi12=0_2} and \eqref{equ:pi-red_2}, we have
\begin{align}
\label{equ:all_zero_2}
\Pi^{\bar{T}}_{123}(\bigl\{\{i\}\bigl\})+ \Pi^{\bar{T}}_{123}(\bigl\{\{i\}\{jk\}\bigl\}) = 1.
\end{align}
Similar to System 1, observe that $H(\bar{T},\bar{S}_1|\bar{S}_2, \bar{S}_3)=0$, i.e., the entire system entropy is provided by~$\bar{S}_2,\bar{S}_3$.
Therefore, all PID atoms that do not include either $\{\bar{S}_2\}$, or
$\{\bar{S}_3\}$, or $\{\bar{S}_2,\bar{S}_3\}$
are zero, i.e.,
\begin{align}\label{equation:Sys1Zeros}
\Pi^{\bar{T}}_{123}(\{\{1\}\}) =   \Pi^{\bar{T}}_{123}(\{\{12\}\})=\Pi^{\bar{T}}_{123}(\{\{13\}\})\nonumber \\=\Pi^{\bar{T}}_{123}(\{\{123\}\})=\Pi^{\bar{T}}_{123}(\{\{12\}\{13\}\})=0.
\end{align}
Taking~\eqref{equ:all_zero_2} with $i=1$, we obtain:
\begin{align*}
\Pi^{\bar{T}}_{123}(\{\{1\}\{23\}\}) =1.
\end{align*}
Next, observing that~$H(\bar{T},\bar{S}_2|\bar{S}_1,\bar{S}_3)=H(\bar{T},\bar{S}_3|\bar{S}_1,\bar{S}_2)=0$, it is readily verified (similar to~\eqref{equation:Sys1Zeros}) that 
\begin{align*}
\Pi^{\bar{T}}_{123}(\{\{2\}\}) =   \Pi^{\bar{T}}_{123}(\{\{12\}\})=\Pi^{\bar{T}}_{123}(\{\{23\}\})\nonumber \\=\Pi^{\bar{T}}_{123}(\{\{123\}\})=\Pi^{\bar{T}}_{123}(\{\{12\}\{23\}\})=0,
\end{align*}
and that
\begin{align*}
\Pi^{\bar{T}}_{123}(\{\{3\}\}) =   \Pi^{\bar{T}}_{123}(\{\{23\}\})=\Pi^{\bar{T}}_{123}(\{\{13\}\})\nonumber \\=\Pi^{\bar{T}}_{123}(\{\{123\}\})=\Pi^{\bar{T}}_{123}(\{\{23\}\{13\}\})=0.
\end{align*}
Finally, taking~\eqref{equ:all_zero_2} with $i=2$ and~$i=3$, it follows that
$\Pi^{\bar{T}}_{123}(\{\{2\}\{13\}\})=\Pi^{\bar{T}}_{123}(\{\{3\}\{12\}\}) =1$.
\,\\

Clearly, System~1 and System~2 above satisfy condition~$(i)$ with respect to the bijection~$\psi$ which maps an antichain in System~1 to an antichain in System~2 having identical indices to all sources (e.g.,~$\psi(\{\{\hat{S}_1\}\}=\{\{\bar{S}_1\}\}$).
Condition~$(ii)$ is also clearly satisfied.
\end{proof}
\end{document}